\documentclass[a4paper,10pt,openright]{article} 
\usepackage[english]{babel} 
\usepackage[italian]{}
\usepackage[latin1]{inputenc} 
\usepackage{amsmath,amssymb, stmaryrd,amsthm,amsfonts,amscd,color,enumerate,eucal,latexsym,mathrsfs,mathtools,cancel,tikz}
\usepackage{epsfig, hieroglf, protosem}  
\usepackage{simplewick, bm}
\usepackage{hyperref}
\usepackage[all,cmtip]{xy}
\usetikzlibrary{matrix,calc}
\usepackage{verbatim}
\numberwithin{equation}{section}

\newtheoremstyle{TheoremStyle}
{3pt}
{3pt}
{\slshape}
{}
{\bf}
{:}
{.5em}
{}

\definecolor{PaColor}{RGB}{255,50, 150}
\definecolor{ClaColor}{RGB}{0,0,255}

\theoremstyle{TheoremStyle}
\newtheorem{theorem}{Theorem}

\newtheorem{proposition}[theorem]{Proposition}
\newtheorem{lemma}[theorem]{Lemma}
\newtheorem{definition}[theorem]{Definition}
\newtheorem{remark}[theorem]{Remark}
\newtheorem{example}[theorem]{Example} 

\title{On a Microlocal Version of Young's Product Theorem}

\author{Claudio Dappiaggi \thanks{CD: Dipartimento di Fisica,
	Universit\`a degli Studi di Pavia \& INFN, Sezione di Pavia, 
	Via Bassi 6,
	I-27100 Pavia,
	Italia; Istituto Nazionale di Alta Matematica, Sezione di Pavia, via Ferrata 5, 27100 Pavia, Italia
	\mbox{claudio.dappiaggi@unipv.it}
}
\and
	Paolo Rinaldi \thanks{PR: Institute for Applied Mathematics, Universit\"at Bonn 
	Endenicher Allee 60,
	D-53115 Bonn,
	Germany;
	 Istituto Nazionale di Alta Matematica, Sezione di Pavia, via Ferrata 5, 27100 Pavia, Italia
	\mbox{rinaldi@iam.uni-bonn.de}
}
\and
	Federico Sclavi \thanks{FS: Dipartimento di Fisica,
	Universit\`a degli Studi di Pavia \& INFN, Sezione di Pavia, 
	Via Bassi 6,
	I-27100 Pavia,
	Italia; Istituto Nazionale di Alta Matematica, Sezione di Pavia, via Ferrata 5, 27100 Pavia, Italia
\mbox{federico.sclavi01@universitadipavia.it}}
}

\begin{document}
\maketitle
\begin{abstract}
A key result in distribution theory is Young's product theorem which states that the product between two H\"older distributions $u\in\mathcal{C}^\alpha(\mathbb{R}^d)$ and $v\in\mathcal{C}^\beta(\mathbb{R}^d)$ can be unambiguously defined if $\alpha+\beta>0$. 
We revisit the problem of multiplying two H\"older distributions from the viewpoint of microlocal analysis, using techniques proper of Sobolev wavefront set. 
This allows us to establish sufficient conditions which allow the multiplication of two H\"older distributions even when $\alpha+\beta\leq 0$.
\end{abstract}


\section{Introduction}

In the framework of the theory of distributions, establishing under which conditions it is possible to define a multiplicative product is a topic which has attracted a lot of attention, especially due to its relevance for applications, see {\it e.g.} \cite{Ober}. 
In the past few years, in view of its relevance in the context of regularity structures and of Hairer's reconstruction theorem \cite{Hai14}, Young's product theorem, addressing the problem of multiplying two tempered distributions which can be read as elements of an H\"older space, has acquired particular relevance.
More precisely, it turns out that, if $u\in\mathcal{C}^\alpha(\mathbb{R}^d)$ and $v\in\mathcal{C}^\beta(\mathbb{R}^d)$ then there exists the product $uv\in\mathcal{C}^{\alpha\wedge\beta}(\mathbb{R}^d)$ with $\alpha\wedge\beta=\min(\alpha,\beta)$ provided that $\alpha+\beta>0$, see {\it e.g.} \cite{BCD11}.

At the same time, the problem of multiplying two distributions can also be addressed from a completely different perspective, namely within the framework of microlocal analysis \cite{Hor03}. 
While this setting has been proficiently used in many contexts, one might wonder whether it can be applied to problem of multiplying two H\"older functions. 
In this paper we address this issue following a specific strategy. 
To start with, we observe that working with the notion of smooth wavefront set as in \cite[Ch. 8]{Hor03} does not allow to catch the specific features of the singular structure of an element lying in a H\"older space. 
In order to bypass this hurdle we consider instead the Besov spaces $B^\alpha_{\infty,\infty}(\mathbb{R}^d)$, $\alpha\in\mathbb{R}$, see \cite{BCD11}, which coincide with $\mathcal{C}^\alpha(\mathbb{R}^d)$ whenever $\alpha\in\mathbb{R}\setminus\mathbb{N}$, while $\mathcal{C}^\alpha(\mathbb{R}^d)\subset B^\alpha_{\infty,\infty}(\mathbb{R}^d)$ if $\alpha\in\mathbb{N}$. 
Recently, a notion of wave-front set encoding Besov regularity has been introduced in \cite{DRS22}.

To start with we discuss the interplay between Besov spaces, and in particular H\"older spaces, and suitable fractional Sobolev spaces. 
This allows us to make use of the concept of Sobolev wave front set \cite{Hor97} which was introduced with the goal of catching a more specific singular behaviour of a distribution in the cotangent space in comparison to the standard smooth wave front set. 

Working in this framework we are able to prove our first version of Young's product theorem. 
In the following with $\mathrm{singsupp}(u)$ we refer to the to the singular support of a distribution $u\in\mathcal{D}^\prime(\mathbb{R}^d)$. In particular, we recall that the singular support of $u$ is the set of points $x\in\mathbb{R}^d$ having no neighbourhood to which the restriction of $u$ is a smooth function \cite[Chapt. 2]{Hor03}.
In addition, given $g \in \mathcal{C}_{\mathrm{loc}}^\beta(\mathbb{R}^d)$ with $\beta<0$ we define
\begin{equation}
\beta^\ast_g \vcentcolon= \sup\{\gamma \leq 0: g \in B_{p,\infty,\mathrm{loc}}^\gamma(\mathbb{R}^d) \text{ for all}\, p\in [2,\infty]\}\,.
\end{equation}
Finally, we call {\em canonical} the product between two distributions, which is defined along the same lines of \cite[Thm 8.2.10]{Hor03}, see also Remark \ref{Rmk: canonicity}.

\begin{theorem}\label{Thm: microlocal Young_intro}
Let $f,g\in\mathcal{D}^\prime(\mathbb{R}^d)$. Suppose that, for any $x\in \operatorname{singsupp}(f)\cap\operatorname{singsupp}(g)$, $\exists\,\varphi_f, \varphi_g \in\mathcal{D}(\mathbb{R}^d)$, with $\varphi_f(x)\neq0$ and $\varphi_g(x)\neq0$, such that $\varphi_ff\in \mathcal{C}_{c}^\alpha(\mathbb{R}^d)$ and  $\varphi_gg\in \mathcal{C}_{c}^\beta(\mathbb{R}^d)$ with $\alpha>0$ and $\beta < 0$ such that
\begin{align}\label{Eq: Novel Young condition intro}
\alpha + \beta^\ast_{\varphi_g g} > 0\,.
\end{align}
Then there exists a canonical product $f \cdot g \in \mathcal{D}^\prime(\mathbb{R}^d)$.
\end{theorem}

\noindent Subsequently we are able to further refine this statement through restricting our attention to H\"older spaces. 
In particular we prove that, when existent, the product between two distributions is a continuous map.

\begin{theorem}\label{Thm: microlocal Young II_intro}
Let $f \in \mathcal{C}_{\mathrm{loc}}^\alpha(\mathbb{R}^d)$ and $g \in \mathcal{C}_{\mathrm{loc}}^\beta(\mathbb{R}^d)$ with $\alpha > 0$ and $\beta < 0$. If $\text{singsupp}(f) \cap \text{singsupp}(g) \neq \emptyset$ and 
\begin{align}\label{Eq: novel Young condition II intro}
\alpha + \beta^\ast_g > 0\,,
\end{align} 
there exists the product $f \cdot g \in \mathcal{C}_{\mathrm{loc}}^\beta(\mathbb{R}^d)$. Moreover for any compact set $K \subset \mathbb{R}^d$, one has
\begin{align}\label{Eq: continuity of microlocal young product intro}
\|f\cdot g\|_{\mathcal{C}^{\beta}(K)}\lesssim\|f\|_{\mathcal{C}^{\alpha}(\overline{K}_1)}\|g\|_{\mathcal{C}^{\beta}(\overline{K}_1)}\,,
\end{align}
where $\overline{K}_1$ denotes the $1$-enlargement of $K$.
\end{theorem}

We observe that the condition $\alpha+\beta>0$ is no longer strictly necessary, as we shall clarify with concrete examples throughout this work.
We observe that both Theorem \ref{Thm: microlocal Young_intro} and Theorem \ref{Thm: microlocal Young II_intro} require the H\"older regularity $\alpha$ of one of the two factors to be positive. 
In Section \ref{Sec: microlocal improvement} we also discuss some scenarios where the product can be (canonically) constructed also if both the regularities $\alpha$ and $\beta$ are negative. 
This is achieved through an extension procedure making use of the notion of \emph{scaling degree}, which we briefly outline in Appendix \ref{App: sd}.

As an application of these results we focus our attention to the recently introduced notion of coherent germs of distributions \cite{CZ20}. 
Most notable it allows to reformulate Hairer's reconstruction theorem from a purely distributional viewpoint and, in turn, to give an alternative derivation of the classical Young's product theorem. 
Here we show that our results can be combined with the framework introduced in \cite{CZ20} to improve also on this notable application of the reconstruction theorem.

\paragraph{Outline of the paper} 
The paper is organized as follows. In Section \ref{Sec: preliminaries} we introduce the function and distribution spaces which are relevant throughout the paper, particularly Besov, H\"older and Sobolev spaces. 
Moreover we discuss the interplay between these spaces which will be among the main ingredients for proving a microlocal version of Young's product theorem using techniques proper of Sobolev wave-front sets, a topic introduced in Section \ref{Sec: Sobolev WF}.

In Section \ref{Sec: microlocal improvement} we present the main results of this paper, see Theorems \ref{Thm: microlocal Young_intro} and \ref{Thm: microlocal Young II_intro}. 
In this section we also discuss some concrete examples of applicability.

In Section \ref{Sec: Application to germs} we discuss an application of the results of Section \ref{Sec: microlocal improvement} to the framework of coherent germs of distributions. 
In particular, we discuss an identification of a canonical reconstruction for a suitable class of germs of distributions in cases where a canonical choice, from the point of view of the sole reconstruction theorem, is missing. 
Finally in Appendix \ref{App: sd} we outline the notion of scaling degree and its main properties, since it plays a key r\^{o}le in our construction.

\paragraph{Notations}
In this short paragraph we fix a few recurring notations of this work. 
More precisely, with $\mathcal{E}(\mathbb{R}^d)$ ({\em resp.} $\mathcal{D}(\mathbb{R}^d))$, we indicate the space of smooth ({\em resp.} smooth and compactly supported) functions, while $\mathcal{S}(\mathbb{R}^d)$ refers to the space of rapidly decreasing smooth functions. Given any function $\psi$ lying in any of these spaces, we shall employ the notation $\psi_x$, $x\in\mathbb{R}^d$ to refer to $\psi_x(y):=\psi(y-x)$.
Their topological dual spaces are denoted respectively $\mathcal{E}^\prime(\mathbb{R}^d)$, $\mathcal{D}^\prime(\mathbb{R}^d)$ and $\mathcal{S}^\prime(\mathbb{R}^d)$. 
In addition, given $u\in\mathcal{S}(\mathbb{R}^d)$, we adopt the following convention to define its Fourier transform
\begin{align*}
\mathcal{F}(u)(\xi)=\widehat{u}(\xi)\vcentcolon=\int_{\mathbb{R}^d}e^{-ik\cdot\xi}u(x)\,dx\,.
\end{align*}
Similarly, for any $v\in\mathcal{S}^\prime(\mathbb{R}^d)$, we indicate with $\widehat{v}\in\mathcal{S}^\prime(\mathbb{R}^d)$ its Fourier transform, defining it per duality as $
\widehat{v}(u)\doteq v(\widehat{u})$ for all $u\in\mathcal{S}(\mathbb{R}^d)$. 
With $\langle x\rangle\vcentcolon=(1+|x|^2)^{\frac{1}{2}}$ we will denote the Japanese bracket, while the symbol $\lesssim$ shall refer to an inequality holding true up to a multiplicative finite constant. 
Finally, given a compact set $K\subset\mathbb{R}^d$, we shall call $R$-enlargement of $K$ 
\begin{align*}
\overline{K}_R\vcentcolon=\{x\in\mathbb{R}^d\;|\;|x-y|\leq R\,,\;\mathrm{for}\;\mathrm{some}\;y\in K\}\,.
\end{align*}
Finally, for any $\alpha\in\mathbb{R}$, we denote $\lfloor \alpha\rfloor\vcentcolon=\max\{n\in\mathbb{Z}:n<\alpha\}$,

\paragraph{Acknowledgements} We are thankful to N. Drago, D. Lee,  M. Wrochna and L. Zambotti for helpful discussions and comments. The work of F.S. is supported by a scholarship of the University of Pavia.

\section{Preliminaries}\label{Sec: preliminaries}
The aim of this section is twofold: On the one hand, we introduce the function and distribution spaces relevant to us, together with the key microlocal structures, in particular the \textit{Sobolev wave-front set}, see \cite{DH72}. 
On the other hand, we discuss some relations between some of these spaces, as preliminary tools to establish a microlocal version of Young's theorem. 

\subsection{Function Spaces}\label{Sec: function spaces}
In this section we give a succinct outline of the main function spaces, relevant for our analysis, referring to the existing literature for further details.

\paragraph{Besov Spaces}
We start from \textit{Besov spaces} \cite{BCD11}. 
To this end, the first step consists of introducing a \textit{Littlewood-Paley} partition of unity, which is a suitable partition of unity in Fourier space. Let
\begin{align*}
\mathcal{O}(\mathbb{R}^d)\vcentcolon=\{f\in\mathcal{E}(\mathbb{R}^d)\,|\,\forall\,\beta\in\mathbb{N},\;\exists\,C,N>0\quad\mathrm{s.t.}\quad|D^\beta f(x)|\leq C\langle x\rangle^N\}\,,
\end{align*}
and let $m(\xi)\in\mathcal{O}(\mathbb{R}^d)$. 
We call \textit{Fourier multiplier operator} the continuous map $m(D):\mathcal{S}^\prime(\mathbb{R}^d)\to\mathcal{S}^\prime(\mathbb{R}^d)$
\begin{align*}
m(D)u\vcentcolon=\mathcal{F}^{-1}\{m(\xi)\widehat{u}(\xi)\}\,,
\end{align*}
where $\mathcal{F}^{-1}$ denotes the inverse Fourier transform operator. 

\begin{definition}\label{Def: L-P partition of unity}
Let $N\in\mathbb{N}$ and let $\psi\in\mathcal{D}(\mathbb{R}^d)$ be a nonnegative function supported in $\{2^{-N}\leq|\xi|\leq2^N\}$. We call \textit{Littlewood-Paley} partition of unity the sequence $\{\psi_j\}_{j\in\mathbb{N}_0}$, $\mathbb{N}_0\doteq\mathbb{N}\cup\{0\}$ such that
\begin{itemize}
\item $\psi_0\in\mathcal{D}(\mathbb{R}^d)$ and $\textrm{supp}(\psi_0)\subseteq\{|\xi|\leq2^N\}$;	
\item $\psi_j(x)\vcentcolon=\psi(2^{-j}x)$ for $j\geq1$;
\item $\sum\limits_{j\in\mathbb{N}_0}\psi_j(\xi)=1$ for all $\xi\in\mathbb{R}^d$;
\item for any multi-index $\alpha$, $\exists C_\alpha>0$ such that
\begin{align*}
|D^\alpha\psi_j(\xi)|\leq C_\alpha\langle\xi\rangle^{-|\alpha|}\,,\qquad j\geq1\,.
\end{align*}
\end{itemize}
\end{definition}
\noindent
In the following we shall always assume $N=1$.

\begin{remark}\label{Rmk: weighted ell^p-spaces}
Given a Banach space $A$, $s\in\mathbb{R}$ and $p\in[1,\infty]$, the set $\ell_p^s(A)$ denotes the space of sequences $a=\{a_k\}_{k\in\mathbb{N}_0}$, $a_k\in A$, such that it is finite
\begin{align*}
\|a\|_{\ell_p^s(A)}\vcentcolon=
\begin{cases}
\bigg(\sum_{k\geq0}(2^{ks}\|a_k\|_A)^p\bigg)^{\frac{1}{p}}\,,\qquad&\mathrm{if}\quad p<\infty\,,\\
\sup_{k\geq0}\,2^{ks}\|a_k\|_A\,,\qquad&\mathrm{if}\quad p=\infty\,,
\end{cases}.
\end{align*}
\end{remark}
\begin{definition}\label{Def: Besov spaces}
Let $\alpha\in\mathbb{R}$, $p,q\in[1,\infty]$ and let $\{\psi_j\}_{j\in\mathbb{N}_0}$ be a Littlewood-Paley partition of unity as per Definition \ref{Def: L-P partition of unity}. 
We call (non-homogeneous) {\bf Besov space} $B^\alpha_{p,q}(\mathbb{R}^d)$ the space of $u\in\mathcal{S}^\prime(\mathbb{R}^d)$ such that $u_\psi=\{\psi_j(D)u\}_{j\in\mathbb{N}_0}\in \ell^s_q(L^p)$, endowed with norm
\begin{align*}
\|u\|_{B^\alpha_{p,q}(\mathbb{R}^d)}\vcentcolon=\|u_\psi\|_{\ell^\alpha_q(L^p)}\,.
\end{align*}
\end{definition}

\noindent The definition does not depend on the choice of the Littlewood-Paley partition of unity.

\begin{remark}\label{Rmk: Besov spaces with p=q=infty}
In the following we shall be particularly interested in the special cases $p=q=\infty$, in particular in $B^\alpha_{\infty,\infty}(\mathbb{R}^d)$ with $\alpha\in\mathbb{R}$, where, for any $u\in B^\alpha_{\infty,\infty}(\mathbb{R}^d)$
\begin{align*}
\|u\|_\alpha\equiv \|u\|_{B^\alpha_{\infty,\infty}(\mathbb{R}^d)}=\sup_{j\geq0}2^{j\alpha}\|\psi_j(D)u\|_{L^\infty(\mathbb{R}^d)}\,.
\end{align*}
\end{remark}

We observe that one can define Besov spaces without resorting to any Littlewood-Paley partition of unity, as in \cite{BL21, HL17}. This fact can be codified in the following Proposition \ref{Def: Besov spaces}, see \cite[Sec. 1.4]{Tri06}.

\begin{proposition}
\label{def:besov}
Let $\alpha \in \mathbb{R}$ and $p,q \in [1,\infty]$. 
Moreover let $r \in \mathbb{N}$ be such that $r > -\alpha$. 
We call $B^\alpha_{p,q}(\mathbb{R}^d)$, $d\geq 1$, the space of $f \in \mathcal{D}^\prime(\mathbb{R}^d)$ such that, for $n\in\mathbb{N}$,
\begin{equation}
\begin{cases}
\bigg\| \bigg\| \sup\limits_{\psi\in\mathcal{B}^r}\bigg\lvert\frac{f(\psi^{2^{-n}}_x)}{2^{-n\alpha}} \bigg\rvert \bigg\|_{L^p(\mathbb{R}^d)} \bigg\|_{\ell^q(\mathbb{N})} < + \infty\,, &\text{if} \quad \alpha < 0\,, \\
\|\sup\limits_{\psi \in \mathcal{B}^r} \lvert f(\psi_x) \rvert\|_{L^p(\mathbb{R}^d)} + \bigg\| \bigg\| \sup\limits_{\psi \in \mathcal{B}^r_{\lfloor \alpha \rfloor}}\bigg \lvert \frac{f(\psi^{2^{-n}}_x)}{2^{-n\alpha}} \bigg\rvert \bigg\|_{L^p(\mathbb{R}^d)} \bigg\|_{\ell^q(\mathbb{N})} < + \infty\,,&\text{if} \quad \alpha \geq 0,
\end{cases}
\end{equation}
where $\mathcal{B}^r$ is the space of test functions $\psi\in \mathcal{D}(\mathbb{R}^d)$ such that $\|\psi\|_{C^r} \leq 1$ while $\mathcal{B}^r_{\lfloor \alpha \rfloor}$ is the space of test functions $\psi \in \mathcal{B}^r$ such that $\int x^k \psi(x) dx = 0$ for any multi-index $k$ such that $0 \leq \lvert k \rvert \leq \lfloor \alpha \rfloor$. 
\end{proposition}

Observe that the previous statement does not depend on the choice of $r > -\alpha$. 
At this stage, it is rather useful to introduce a local version of the spaces $B^\alpha_{p,q}(\mathbb{R}^d)$. 
Locality is established requiring the bounds of Definition \eqref{def:besov} to hold true with $L^p(\mathbb{R}^d)$ being replaced by $L^p(K)$ for any compact set $K\subset \mathbb{R}^d$.

\begin{proposition}
\label{def:besov_loc}
Let $\alpha \in \mathbb{R}$ and $p,q \in [1,\infty]$.
Moreover let $r \in \mathbb{N}$ such that $r > -\alpha$. 
We call $B^\alpha_{p,q,\mathrm{loc}}(\mathbb{R}^d)$, $d\geq 1$, the space of distributions $f \in \mathcal{D}^\prime(\mathbb{R}^d)$ such that for all compact sets $K \subset \mathbb{R}^d$, for $n\in\mathbb{N}$,
\begin{equation}\label{Eq: Local Besov norms}
	\begin{cases}
		\bigg\| \bigg\| \sup\limits_{\psi \in \mathcal{B}^r}\bigg \lvert \frac{f(\psi^{2^{-n}}_x)}{2^{-n\alpha}} \bigg\rvert \bigg\|_{L^p(K)} \bigg\|_{\ell^q(\mathbb{N})} < + \infty\,,& \text{if} \quad \alpha < 0\,, \\
		\|\sup\limits_{\psi \in \mathcal{B}^r} \lvert f(\psi_x) \rvert\|_{L^p(K)} + \bigg\| \bigg\| \sup\limits_{\psi \in \mathcal{B}^r_{\lfloor \alpha \rfloor}}\bigg \lvert \frac{f(\psi^{2^{-n}}_x)}{2^{-n\alpha}} \bigg\rvert \bigg\|_{L^p(K)} \bigg\|_{\ell^q(\mathbb{N})} < + \infty\,,&\text{if}\quad\alpha\geq0\,.
	\end{cases}
\end{equation}
\end{proposition}

\noindent We recall some further characterizations of Besov spaces which we shall exploit later. 
For the proof of this result we refer to \cite{BL21}.

\begin{theorem}
\label{th:equivalent_norms}
Let $\alpha \in \mathbb{R}$ and $p,q\in [1,+\infty]$. 
Then, using the same notation and nomenclature as in Definition \ref{def:besov}, the following statements hold true:
\begin{itemize}
\item If $\alpha < 0$, then $f \in B^{\alpha}_{p,q}(\mathbb{R}^d)$ if and only if $f\in\mathcal{D}^\prime(\mathbb{R}^d)$ and $\exists\varphi \in \mathcal{D}(\mathbb{R}^d)$ with $\int\limits_{\mathbb{R}^d}dx\, \varphi(x) \neq 0$ such that
\begin{equation}
\label{eq:negative regularity besov}
\bigg\| \bigg\| \frac{f(\varphi^{2^{-n}}_x)}{2^{-n\alpha}} \bigg\|_{L^p(\mathbb{R}^d)} \bigg\|_{\ell^q(\mathbb{N})} < +\infty\,;
\end{equation}
\item If $\alpha > 0$ and $\alpha \not \in \mathbb{N}$, then $f \in B^{\alpha}_{p,q}(\mathbb{R}^d)$ if and only if $f\in\mathcal{D}^\prime(\mathbb{R}^d)$ and for every multi-index $k$ such that $0 \leq \lvert k \rvert < \alpha$, $\partial^k f \in L^p(\mathbb{R}^d)$ and for any $h_0 > 0$,
\end{itemize}
\begin{equation}
\bigg\| \bigg\| \frac{\partial^k f(x+h) - \sum_{0 \leq \lvert \ell \rvert < \alpha - \lvert k \rvert} \frac{1}{\ell!} \partial^{k+\ell}f(x) h^\ell}{h^{\alpha - \lvert k \rvert}} \bigg \|_{L^p(\mathbb{R}^d)} \bigg \|_{L_h^q(B(0,h_0))} < +\infty\,,
\end{equation}
where $h\in B(0,h_0)$, while $L_h^q(B(0,h_0)):=L^q\bigg(B(0,h_0), \frac{dh}{\lvert h \rvert^d}\bigg)$.
\end{theorem}

\begin{remark}
A counterpart of Theorem \ref{th:equivalent_norms} holds true also for local Besov spaces, provided that the $L^p(\mathbb{R}^d)$-norms are replaced by the $L^p(K)$-norms for every compact $K\subset \mathbb{R}^d$. 
\end{remark}

\begin{remark}
Observe that, as proven in \cite[Corollary 1.12]{Tri06}, if Equation \eqref{eq:negative regularity besov} holds true for one $\varphi \in\mathcal{D}(\mathbb{R}^d)$ with $\int\limits_{\mathbb{R}^d}dx\, \varphi(x) \neq 0$, then it must hold true for all such functions.
\end{remark}

\subparagraph{H\"older Spaces}
As anticipated in Remark \ref{Rmk: Besov spaces with p=q=infty}, in the following we shall be interested in a distinguished class of Besov spaces, namely $B^\alpha_{\infty,\infty}(\mathbb{R}^d)$ with $\alpha\in\mathbb{R}$, since, when $\alpha\notin\mathbb{N}$ they can be identified as H\"older spaces as a consequence of Theorem \ref{th:equivalent_norms}. These can be defined independently, separating the cases $\alpha> 0$ from those with $\alpha< 0$, see also \cite{KNVW08}. 
We start from the former.
\begin{definition}\label{Def: C-alpha-loc}
Let $\alpha>0$. 
The space $\mathcal{C}^\alpha_{\mathrm{loc}}(\mathbb{R}^d)$ of locally $\alpha$-H\"{o}lder functions  consists of all $f\in C^{\lfloor\alpha\rfloor}(\mathbb{R}^d)$ such that, for any compact set $K\subset\mathbb{R}^d$,
\begin{align*}
\lvert f(y) - P_x(y)\rvert\lesssim\lvert x-y\rvert^\alpha\,,
\end{align*}
uniformly for $x,y \in K$, where $P_x(\cdot)$ is the Taylor polynomial of order $\lfloor \alpha\rfloor$ of $f$ centered at $x$, \textit{i.e.},
\begin{align*}
P_x(y) \vcentcolon=\sum_{\lvert k\rvert<\alpha}\partial^kf(x)\frac{(y-x)^k}{k!}\,, \qquad y\in\mathbb{R}^d\,.
\end{align*}
Moreover, for all $\alpha\in\mathbb{R}_+$, we endow $\mathcal{C}^\alpha_{\mathrm{loc}}(\mathbb{R}^d)$ with the following seminorms: 
\begin{align*}
\|f\|_{\mathcal{C}^\alpha(K)}\vcentcolon=\max\bigg\{\|f\|_{C^{\lfloor\alpha\rfloor}(K)}\,,\sup_{x,y \in K}\frac{\lvert f(y)-P_x(y)\rvert}{\lvert x-y\rvert^\alpha}\bigg\}\,,
\end{align*}
where $K$ runs over all compact sets in $\mathbb{R}^d$.
\end{definition}

\noindent If we focus the attention on $\alpha<0$ the following definition holds true.

\begin{definition}\label{Def: C-alpha-loc-negativi}
Let $\alpha\in(-\infty,0)$, we call $\mathcal{C}^\alpha_{\mathrm{loc}}(\mathbb{R}^d)$ the space of distributions $u\in\mathcal{D}^\prime(\mathbb{R}^d)$ such that, for any compact set $K\subset\mathbb{R}^d$,
\begin{align}\label{Eq: Holder characterization of negative regularity}
\lvert u(\varphi^\lambda_x) \rvert \lesssim \lambda^\alpha\,,
\end{align}
uniformly for $x,y \in K$, for $\lambda \in (0,1]$ and for $\varphi \in \mathcal{B}_{-\lfloor \alpha \rfloor}$, where $\mathcal{B}_{-\lfloor \alpha \rfloor}:=\{\phi \in \mathcal{D}(B(0,1)) : \|\phi\|_{C^{-\lfloor \alpha \rfloor}(\mathbb{R}^d)} \leq 1\}$. Here $\varphi^\lambda_x(y):=\lambda^{-d}\varphi(\frac{y-x}{\lambda})$.
Moreover we endow $\mathcal{C}^\alpha_{\mathrm{loc}}(\mathbb{R}^d)$ with the following seminorms:
\begin{align*}
\|u\|_{\mathcal{C}^\alpha(K)} := \sup_{x\in K, \varphi \in \mathcal{B}_{-\lfloor \alpha \rfloor}, \lambda \in (0,1]} \frac{\lvert u(\varphi^\lambda_x) \rvert}{\lambda^\alpha}\,.
\end{align*}
\end{definition}

\noindent To conclude, on account of Theorem \ref{th:equivalent_norms}, we observe that
\begin{itemize}
	\item $\mathcal{C}^\alpha(\mathbb{R}^d)\equiv B^\alpha_{\infty,\infty}(\mathbb{R}^d)$ if $\alpha\notin\mathbb{N}$ while $\mathcal{C}^\alpha(\mathbb{R}^d)\subset B^\alpha_{\infty,\infty}(\mathbb{R}^d)$ if $\alpha\in\mathbb{N}$;
	\item $B^\alpha_{\infty,\infty, c}(\mathbb{R}^d)\vcentcolon=B^\alpha_{\infty,\infty,\mathrm{loc}}(\mathbb{R}^d)\cap\mathcal{E}^\prime(\mathbb{R}^d)$ while $\mathcal{C}^\alpha_{c}(\mathbb{R}^d)\vcentcolon=\mathcal{C}^\alpha_{\mathrm{loc}}(\mathbb{R}^d)\cap\mathcal{E}^\prime(\mathbb{R}^d)$.
\end{itemize}



\paragraph{Fractional Sobolev Spaces}
In this paragraph we recall the notion of fractional Sobolev spaces.

\begin{definition}
Let $s\in\mathbb{R}$. We call {\bf Sobolev space (of order $s$)} 
$$H^s(\mathbb{R}^d):=\{ u\in\mathcal{S}^\prime(\mathbb{R}^d)\;|\;
\langle D\rangle^su\in L^2(\mathbb{R}^d)\},$$
where $\langle D\rangle^su\vcentcolon=\mathcal{F}^{-1}\{\langle\xi\rangle^s\widehat{u}(\xi)\}$.
Moreover we say that $u\in H_{\mathrm{loc}}^s(\mathbb{R}^d)$ if, for any $\varphi\in\mathcal{D}(\mathbb{R}^d)$, $\varphi u\in H^s(\mathbb{R}^d)$.
\end{definition}

\begin{remark}\label{Rmk: Fourier Plancherel}
We recall that, on account of Fourier-Plancherel theorem \cite{Hor03}, a distribution  $u\in\mathcal{S}^\prime(\mathbb{R}^d)$ lies in $H^s(\mathbb{R}^d)$ if and only if
\begin{align*}
\int_{\mathbb{R}^d}\langle\xi\rangle^{2s}\lvert\widehat{u}(\xi)\rvert^2\,d\xi<\infty\,.
\end{align*}
\end{remark}

\paragraph{Embedding in Fractional Sobolev Spaces} In this paragraph, we prove two theorems concerning the embedding of Besov spaces into suitable fractional Sobolev spaces.
In order to smoothen the reading, we consider separately the case of positive and negative regularity.
These two theorems play a key r\^ole in the proof of the microlocal version of Young's product theorem.

\begin{theorem}
\label{th:inclusions_negative_regularity}
Let $\alpha < 0$. Then 
\begin{itemize}
\item[(a)]If $p\geq 2$ and $q > 2$, $B^\alpha_{p,q,\mathrm{loc}}(\mathbb{R}^d) \subset H^s_{\mathrm{loc}}(\mathbb{R}^d)$ for any $s < \alpha$.
\item[(b)] If $p\geq 2$ and $q = 2$, $B^\alpha_{p,q,\mathrm{loc}}(\mathbb{R}^d) \subset H^s_{\mathrm{loc}}(\mathbb{R}^d)$ for any $s \leq \alpha$,
where $B^\alpha_{p,q,\mathrm{loc}}(\mathbb{R}^d)$ is as per Definition \ref{def:besov_loc}.
\end{itemize}
\end{theorem}

\begin{proof}
We prove the two items separately
\begin{itemize}
\item[(a)] We show that, if $f \in B^\alpha_{p,q,loc}(\mathbb{R})$, then  $f \in H^{s,loc}(\mathbb{R}^d)$ for any $s < \alpha$. Since $H^s(\mathbb{R}^d)\simeq B^s_{2,2}(\mathbb{R}^d)$, see \cite[Sec. 2.7]{BCD11}, in view of Theorem \ref{th:equivalent_norms}, the sought statement descends if, denoting by $K \subset \mathbb{R}^d$ an arbitrary, but fixed compact set, we exhibit $\varphi \in \mathcal{D}(\mathbb{R}^d)$ with $\int\limits_{\mathbb{R}^d}dx\, \varphi(x) \neq 0$ such that $\|f\|_{B^s_{2,2}(K)}<\infty$. Here $\|\cdot\|_{B^s_{2,2}(K)}$ is a shortcut to indicate the norm as per Equation \eqref{Eq: Local Besov norms}.
Still on account of Theorem \ref{th:equivalent_norms}, we know that since $f \in B^\alpha_{p,q,\mathrm{loc}}(\mathbb{R}^d)$ there exists $\varphi^\prime \in \mathcal{D}(\mathbb{R}^d)$ with $\int\limits_{\mathbb{R}^d}dx\, \varphi^\prime(x) \neq 0$ such that $\|f\|_{B^\alpha_{p,q}(K)}<\infty$. 
Making the ansatz that $\varphi^\prime$ plays the r\^{o}le of $\varphi$,
\begin{align}
&\|f\|^2_{B_{2,2}^s(K)}=\sum_{n \geq 0} 2^{2ns} \|f(\varphi^{2^{-n}}_x)\|^2_{L^2(K)} \lesssim_K \sum_{n \geq 0} 2^{2ns} \|f(\varphi^{2^{-n}}_x)\|^2_{L^p(K)} =\nonumber \\
&=\sum_{n \geq 0} \underbrace{2^{2n(s-\alpha)}}_{\ell^{(\frac{q}{2})^\prime}(\mathbb{N})} \underbrace{2^{2n\alpha} \|f(\varphi^{2^{-n}}_x)\|^2_{L^p(K)}}_{\ell^{\frac{q}{2}}(\mathbb{N})}\lesssim_K \|2^{2n(s-\alpha)}\|_{\ell^{(\frac{q}{2})^\prime}(\mathbb{N})} \|2^{2n\alpha} \|f(\varphi^{2^{-n}}_x)\|^2_{L^p(K)}\|_{\ell^{\frac{q}{2}}(\mathbb{N})} =\nonumber \\
&= \|f\|^2_{B^\alpha_{p,q}(K)} \|2^{2n(s-\alpha)}\|_{\ell^{(\frac{q}{2})^\prime}(\mathbb{N})} < +\infty\,. \nonumber
\end{align}
In the first bound, we applied the hypothesis $p\geq 2$ in order to exploit the continuous inclusions $L^p(\Omega) \subset L^2(\Omega)$ for any $p\geq 2$ where $\Omega\subset \mathbb{R}^d$ is a bounded domain. 
In the second bound, we used H\"{o}lder inequality for sequences thanks to the condition $q > 2$, which allows to consider sequences belonging to $\ell^r(\mathbb{N})$ with $r > 1$.
 As a consequence, setting $(\frac{q}{2})^\prime$ to be such that $\frac{1}{(\frac{q}{2})}+\frac{1}{(\frac{q}{2})^\prime}=1$, we conclude that $\|2^{2n(s-\alpha)}\|_{\ell^{(\frac{q}{2})^\prime}(\mathbb{N})}$ is finite for any $s < \alpha$.
\item[(b)] If $p \geq 2$ and $q = 2$, the thesis descends following the same steps as in point (a), the only notable difference being the presence of the $\ell^\infty$-norm of $2^{2n(s-\alpha)}$, which is finite.
\end{itemize}
\end{proof}

\begin{theorem}
\label{th:inclusions_positive_regularity}
Let $\alpha \geq 0$. Then 
\begin{itemize}
\item[(a')]If $p\geq 2$ and $q > 2$, $B^\alpha_{p,q,\mathrm{loc}}(\mathbb{R}^d) \subset H^s_{\mathrm{loc}}(\mathbb{R}^d)$ for any $s < \alpha$.
\item[(b')] If $p\geq 2$ and $q = 2$, $B^\alpha_{p,q,\mathrm{loc}}(\mathbb{R}^d) \subset H^s_{\mathrm{loc}}(\mathbb{R}^d)$ for any $s \leq \alpha$,
\end{itemize}
where $B^\alpha_{p,q,\mathrm{loc}}(\mathbb{R}^d)$ is as per Definition \ref{def:besov_loc}.
\end{theorem}

\begin{proof}
		
We prove the two items separately and we employ the same notation and nomenclature as in Definition \ref{def:besov}.
\begin{itemize}
\item[(a')] Let $f \in B^\alpha_{p,q,\mathrm{loc}}(\mathbb{R}^d)$ with $\alpha > 0$ and let $r > -s$. Letting $K \subset \mathbb{R}^d$ be an arbitrary but fixed compact set, we show that $\|\sup\limits_{\psi \in \mathcal{B}^r} \lvert f(\psi_x) \rvert\|_{L^2(x\in K)}$ and $\bigg\| \bigg\| \sup\limits_{\psi \in \mathcal{B}^r_{\lfloor s \rfloor}}\bigg \lvert \frac{f(\psi^{2^{-n}}_x)}{2^{-ns}} \bigg\rvert \bigg\|_{L^2(K)} \bigg\|_{\ell^2(\mathbb{N})}$ are bounded for $s < \alpha$. Starting from the former, it holds
\[
\|\sup\limits_{\psi \in \mathcal{B}^r} \lvert f(\psi_x) \rvert\|_{L^2(K)} \lesssim_K \|\sup_{\psi \in \mathcal{B}^r} \lvert f(\psi_x) \rvert\|_{L^p(K)} \leq \|f\|_{B^\alpha_{p,q}(K)} < +\infty,
\] 
where, in the inequality $\lesssim_K$, we used that $L^p(K) \subset L^2(K)$ for any $p\geq 2$. Focusing instead on the latter, it descends
\begin{align*}
&\sum_{n \geq 0} 2^{2ns} \| \sup\limits_{\psi \in \mathcal{B}^r_{\lfloor s \rfloor}}\lvert f(\psi^{2^{-n}}_x) \rvert \|^2_{L^2(K)} \overset{\mathcal{B}^r_{\lfloor s \rfloor} \subset \mathcal{B}^r_{\lfloor \alpha \rfloor}}{\leq} \sum_{n \geq 0} 2^{2ns} \| \sup\limits_{\psi \in \mathcal{B}^r_{\lfloor \alpha \rfloor}}\lvert f(\psi^{2^{-n}}_x) \rvert \|^2_{L^2(K)}\lesssim_K \nonumber \\ 
&\lesssim_K \sum_{n \geq 0} 2^{2ns} \| \sup\limits_{\psi \in \mathcal{B}^r_{\lfloor \alpha \rfloor}}\lvert f(\psi^{2^{-n}}_x) \rvert \|^2_{L^p(K)} 
= \sum_{n \geq 0} \underbrace{2^{2n(s-\alpha)}}_{\ell^{(\frac{q}{2})^\prime}(\mathbb{N})} \underbrace{2^{2n\alpha} \| \sup\limits_{\psi \in \mathcal{B}^r_{\lfloor \alpha \rfloor}}\lvert f(\psi^{2^{-n}}_x) \rvert \|^2_{L^p(K)}}_{\ell^{\frac{q}{2}}(\mathbb{N})}\lesssim_K \nonumber \\
&\lesssim_K \|2^{2n(s-\alpha)}\|_{\ell^{(\frac{q}{2})^\prime}(\mathbb{N})} \|\|\sup\limits_{\psi \in \mathcal{B}^r_{\lfloor \alpha \rfloor}}\lvert f(\psi^{2^{-n}}_x) \rvert \|^2_{L^p(K)}\|_{\ell^{\frac{q}{2}}} 
\leq \|f\|^2_{B^\alpha_{p,q}(K)} \|2^{2n(s-\alpha)}\|_{\ell^{(\frac{q}{2})^\prime}(\mathbb{N})} \overset{s <\alpha}{<} \infty\,,
\end{align*}
where once more $(\frac{q}{2})^\prime$ is such that $\frac{1}{(\frac{q}{2})}+\frac{1}{(\frac{q}{2})^\prime}=1$ and where we have used that $L^p(K) \subset L^2(K)$ for any $p\geq 2$ and the H\"{o}lder inequality. Observe that with slight modifications in the previous argument one can prove the same statement in the case $\alpha=0$, setting where necessary $\lfloor 0\rfloor=-1$ and $\mathcal{B}^r_{-1} = \mathcal{B}^r$ \cite{HL17}. \\

\item[(b')] The statement follows tracking the same steps as in item (a'), with the notable exception that we need to evaluate the $\ell^\infty$-norm of $2^{2n(s-\alpha)}$.
\end{itemize}
\end{proof}

\subsection{Sobolev Wave Front Set}\label{Sec: Sobolev WF}
In this section we report succinctly the definition and the main properties of the \textit{Sobolev wave-front set} \cite{DH72, JS02, Hor85, Hor97}. 
Its main merit is to codify the directions in Fourier space where, at a given point of its support, a distribution does not have a suitable fractional Sobolev regularity. 
Together with Theorems \ref{th:inclusions_negative_regularity} and \ref{th:inclusions_positive_regularity}, this will be the main tool for our microlocal version of Young theorem. 
Throughout this section our main reference is \cite[App. B]{JS02}, whereas, for the notion of smooth wave front set we refer to \cite[Ch. 8]{Hor03}.

\begin{definition}\label{Def: Sobolev WF}
Let $X$ be an open subset of $\mathbb{R}^d$ and let $u\in\mathcal{D}^\prime(X)$. 
We say that $(x_0,\xi_0)\in X\times\mathbb{R}^d\setminus\{0\}$ is not in $\mathrm{WF}^s(u)$, the $H^s$-wavefront set of $u$, $s\in\mathbb{R}$, if $\exists\varphi\in\mathcal{D}(X)$ with $\varphi(x_0)\neq0$ and an open conic neighborhood $\Gamma$ of $\xi_0$ such that
\begin{equation}
\int_\Gamma\langle\xi\rangle^{2s}\lvert\widehat{\varphi u}(\xi)\rvert^2\,d\xi<\infty\,.
\end{equation}
\end{definition}

\begin{remark}\label{Rmk: Sobolev regularity and Sobolev WF}
We observe that, on account of Definition \ref{Def: Sobolev WF}, $\mathrm{WF}^s(u) = \emptyset$ if and only if $u \in H^s_{\mathrm{loc}}(X)$. 
Thus, recalling Theorems \ref{th:inclusions_negative_regularity} and \ref{th:inclusions_positive_regularity}, if $u\in \mathcal{C}^\alpha_{c}(\mathbb{R}^d)$ with $\alpha\in\mathbb{R}$, $WF^s(u)=\emptyset$ for any $s<\alpha$.
\end{remark}

In the spirit of microlocal analysis, the Sobolev wave front set yields sufficient conditions to give meaning to certain, otherwise ill-defined structures, for those distributions with a suitable microlocal behaviour. 
In particular, we shall report here the results on the tensor product, on the pointwise product and on the pull-back of distributions.  

\begin{proposition}[\cite{JS02}, Prop B.5]
Let $X,Y\subseteq\mathbb{R}^d$ be open sets and let $u\in\mathcal{D}^\prime(X)$, while $v\in\mathcal{D}^\prime(Y)$. 
Then $\exists u\otimes v\in \mathcal{D}^\prime(X \times Y)$ and
\begin{align} 
\mathrm{WF}^r(u \otimes v)\subset&\mathrm{WF}^s(u)\times\mathrm{WF}(v)\cup\mathrm{WF}(u)\times\mathrm{WF}^t(v) \nonumber \\
\cup&
\begin{cases}
\big((\operatorname{supp}(u)\times\{0\})\times\mathrm{WF}(v)\big)\cup\big(\mathrm{WF}(u)\times(\operatorname{supp}(v)\times\{0\})\big)\,,\quad&r=s+t\,,\\
\big((\operatorname{supp}(u)\times\{0\})\times\mathrm{WF}^t(v)\big)\cup\big(\mathrm{WF}^s(u)\times(\operatorname{supp}(v)\times\{0\})\big)\,,\quad &r=\min\{s,t,s+t\}\,,
\end{cases}
\end{align}
where $\mathrm{WF}$ is the $C^\infty$-wave front set.
\end{proposition}

\begin{proposition}[\cite{JS02}, Prop B.6 and \cite{Ober86}]\label{Prop: Sobolev product of distributions}
Let $X\subseteq\mathbb{R}^d$ be an open set and let $u,v\in\mathcal{D}^\prime(X)$. 
If for any $(x,\xi)\in X \times\mathbb{R}^d\setminus\{0\}$,  $\exists s_1,s_2\in\mathbb{R}$ with $s_1+s_2\geq0$ such that $(x,\xi)\notin\mathrm{WF}^{s_1}(u)$ and $(x,-\xi)\notin\mathrm{WF}^{s_2}(v)$, the product $u\cdot v\in\mathcal{D}^\prime(X)$ exists.
\end{proposition}

\begin{remark}\label{Rmk: pull-back definition}
Since we shall make use of the notion of pull-back of a distribution we recall its definition, \cite[Thm. 6.1.2]{Hor03}: Let $U,V\subseteq\mathbb{R}^d$ be open sets and let $f:U\to V$ be a smooth map such that $df_x$ is surjective for any $x\in U$, with $df_x$ denoting the differential of $f$ at $x$. The pull-back of a distribution via $f$ is the unique continuous and linear map $f^*:\mathcal{D}^\prime(V)\to\mathcal{D}^\prime(U)$ such that $f^*u=u\circ f$ when $u\in C^0(V)$ is a continuous function.
\end{remark}

To conclude, we consider the restriction of distributions to submanifolds. 
In particular, let $Y$ be a $(d-1)$-dimensional manifold together with a smooth embedding $f:Y\to X\subseteq\mathbb{R}^n$.
We denote with $Y_f=f[Y]$ and with $N_f$ the conormal bundle of $Y$, \textit{i.e.},
\begin{align*}
N_f\vcentcolon=\{(f(y),\xi)\in X\times\mathbb{R}^d\,|\,y\in Y\,,\; df^t_y(\xi)=0\}\,.
\end{align*}
\noindent

\begin{proposition}[\cite{JS02}, Prop B.7]\label{Prop: conormal bundle}
With the above notations, let $u\in\mathcal{D}^\prime(X)$ and $s>1/2$ be such that $\mathrm{WF}^s(u)\cap N_f = \emptyset$. 
Then there exists the restriction $u_{Y_f}\in\mathcal{D}^\prime(f[Y])$ such that
	\begin{align*}
	u_{Y_f}(\varphi)=u\cdot(\varphi\delta_{f[Y]})(1)\,,\qquad\forall\,\varphi\in\mathcal{D}(f[Y])\,,
\end{align*}
where $\cdot$ denotes the pointwise product and where $\varphi\delta_{f[Y]}:\mathcal{E}(X)\to\mathbb{C}$ is the distribution such that $(\varphi\delta_{f[Y]})(\psi)=\int_{f[Y]}\varphi(y)\psi(y)\,dy$ where $dy$ is the pull-back to $f[Y]$ of the Lebesgue measure on $\mathbb{R}^d$. In addition we can identify $u_Y\in\mathcal{D}^\prime(Y)$ such that $u_Y:=f^*u_{Y_f}$ and
\begin{align*}
\mathrm{WF}^{s-1/2}(u_Y)\subseteq f^\ast\mathrm{WF}^s(u) \vcentcolon = \{(y,df^t_y(\xi))\,|\, (f(y),\xi)\in\mathrm{WF}^s(u)\}\,,
\end{align*}
for any $s>1/2$.
\end{proposition}

\begin{remark}
Under the assumptions of Proposition \ref{Prop: conormal bundle}, consider a coordinate system $\{x_i\}_{i=1\dots d}$ on $\mathbb{R}^d$ such that $f[Y]$ coincides with the locus	$x_1=0$. Then $\varphi\delta_{f[Y]}$ can be written as $\varphi(x_2,\ldots,x_d)\delta_0(x_1)$ and it follows
	\begin{align*}
		WF^s(\varphi\delta_{f[Y]}) \subseteq
		\begin{cases}
			\emptyset\,, \quad&\mathrm{if}\;\,s<-\frac{1}{2}\,,  \\
			N_f\,, \quad&\mathrm{if}\;\,s\geq-\frac{1}{2}\,.
		\end{cases}
	\end{align*}
\end{remark}

\begin{remark}\label{Rmk: product as pull-back}
We observe that, under the hypotheses of Proposition \ref{Prop: Sobolev product of distributions}, Proposition \ref{Prop: conormal bundle} entails that the pointwise product $u\cdot v$ can be realized as the pull-back $\Delta^*(u\otimes v)$ where $\Delta:X\ni x\mapsto(x,x)\in X\times X$ is the diagonal map.
\end{remark}
\section{Microlocal Version of Young Theorem}\label{Sec: microlocal improvement}
In this section we discuss a reformulation of Young's theorem using microlocal techniques. As a premise, we briefly recall the original statement of this theorem, which is a direct consequence of Theorem \cite[Thm. 2.82]{BCD11} applied to $B^\alpha_{\infty,\infty,\mathrm{loc}}(\mathbb{R}^d)$, with $\alpha\in\mathbb{R}$ (see also \cite{Hai14, CZ20}). 

\begin{theorem}[Young's Product Theorem]\label{Thm: classical Young}
Let $\alpha,\beta\in\mathbb{R}$ be such that $\alpha+\beta>0$. 
Then the map $(f,g)\mapsto f \cdot g$ extends to a continuous bilinear map from $B^\alpha_{\infty,\infty,\mathrm{loc}}(\mathbb{R}^d)\times B^\beta_{\infty,\infty,\mathrm{loc}}(\mathbb{R}^d)$ to $B^{\alpha\wedge\beta}_{\infty,\infty,\mathrm{loc}}(\mathbb{R}^d)$ where $\alpha\wedge\beta=\min(\alpha,\beta)$. 
\end{theorem}

\begin{remark}\label{Rmk: canonicity}
In the following analyses, we will be interested in discussing the existence of a product between suitable classes of distributions. In order to highlight that our construction can be seen as a natural generalization to the case in hand of the classical result proven in \cite[Thm. 8.2.10]{Hor03}, we shall call {\em canonical} the product $f\cdot g$ of two suitably chosen distributions provided that $f\cdot g=\Delta^*(f\otimes g)$ where $\Delta^*$ is the pull-back along the diagonal map $\Delta:x\mapsto(x,x)$, see also Remark \ref{Rmk: pull-back definition}.  In addition this characterization of the product between two distributions in terms of a pull-back is unique, following slavishly the same argument used in \cite[Thm. 8.2.4]{Hor03}, which can be adapted slavishly to the cases in hand.
\end{remark}

In the following, given $u\in\mathcal{D}^\prime(\mathbb{R}^d)$, we shall denote $\mathrm{singsupp}(u)$ the singular support of $u$, \cite{Hor03}. We prove a result for the product of H\"older functions which extends Theorem \ref{Thm: classical Young}.

In addition, Let $g \in \mathcal{C}_{\mathrm{loc}}^\beta(\mathbb{R}^d)$ with $\beta<0$ and let
\begin{equation}\label{Eq: beta*}
\beta^\ast_g \vcentcolon= \sup\{\gamma \leq 0: g \in B_{p,\infty,\mathrm{loc}}^\gamma(\mathbb{R}^d) \text{ for all}\, p\in [2,\infty]\}\,.
\end{equation}

\begin{example}
Let $\delta$ be the Dirac delta centered at the origin. 
Since $\delta \in B^{-d+\frac{d}{p}}_{p,\infty}(\mathbb{R}^d)$ for any $p\in [1,\infty]$, then $\beta^\ast_\delta \geq -\frac{d}{2}$. 
Actually $-d/2$ is the maximum of the set appearing on the right hand side of \eqref{Eq: beta*}
Indeed, suppose that there exist $p> 2$ and $\alpha > -d/2$ such that $\delta \in B^\alpha_{p,\infty}(\mathbb{R}^d)$. 
On account of the embedding theorem for Besov spaces -- \textit{cf.} \cite[Prop. 2.71]{BCD11}, it holds
\[
B^\alpha_{p,\infty}(\mathbb{R}^d) \hookrightarrow B^{\alpha-\frac{d}{p}}_{\infty,\infty}(\mathbb{R}^d)\hookrightarrow B^{-d}_{\infty,\infty}(\mathbb{R}^d)\,,
\]
as a consequence of $\alpha -d/p > -d$. 
Then the Dirac delta distributions lies in a H\"{o}lder space $\mathcal{C}^\gamma(\mathbb{R}^d)$ with a better regularity than $-d$. This is a contradiction and we conclude that $\beta^\ast_\delta = -\frac{d}{2}$.
\end{example}

\begin{example}
Let $g \in \mathcal{C}^\beta(\mathbb{R}^d)$ be $g(x) \vcentcolon=\lvert x\rvert^\beta$ with $\beta < - d/2$. 
It turns out that $g \in B_{2,\infty}^{\beta+\frac{d}{2}}(\mathbb{R}^d)$. 
Then $\beta^\ast_g \geq \beta+\frac{d}{2}$. 
In the same spirit of the previous example, we prove that $\beta + d/2$ is actually a maximum. 
Suppose there exist $p> 2$ and $\alpha > \beta + d/2$ such that $g \in B^\alpha_{p,\infty}(\mathbb{R}^d)$. 
On account of the embedding theorem for Besov spaces -- \textit{cf.} \cite[Prop. 2.71]{BCD11}, it holds
\[
B^\alpha_{p,\infty}(\mathbb{R}^d) \hookrightarrow B^{\alpha-\frac{d}{p}}_{\infty,\infty}(\mathbb{R}^d)\hookrightarrow B^{\beta}_{\infty,\infty}(\mathbb{R}^d)\,,
\]
since $\alpha -d/p > \beta$. 
This $g$ belongs to a H\"{o}lder space $\mathcal{C}^\gamma(\mathbb{R}^d)$ with $\gamma > \beta$. This is not possible because $g \in \mathcal{C}^\beta(\mathbb{R}^d)$ and we conclude that $\beta^\ast_g = \beta + \frac{d}{2}$.
\end{example}

\begin{theorem}\label{Thm: microlocal Young}
Let $f,g\in\mathcal{D}^\prime(\mathbb{R}^d)$. Suppose that, for any $x\in \operatorname{singsupp}(f)\cap\operatorname{singsupp}(g)$, $\exists\,\varphi_f, \varphi_g \in\mathcal{D}(\mathbb{R}^d)$, with $\varphi_f(x)\neq0$ and $\varphi_g(x)\neq0$, such that $\varphi_ff\in \mathcal{C}_{c}^\alpha(\mathbb{R}^d)$ and  $\varphi_gg\in \mathcal{C}_{c}^\beta(\mathbb{R}^d)$ with $\alpha>0$ and $\beta < 0$ such that
\begin{align}\label{Eq: Novel Young condition}
\alpha + \beta^\ast_{\varphi_g g} > 0\,.
\end{align}
Then there exists a canonical product $f \cdot g \in \mathcal{D}^\prime(\mathbb{R}^d)$.
\end{theorem}

\begin{proof}
Let $x\in \operatorname{singsupp}(f)\cap\operatorname{singsupp}(g)$ and let $\varphi_f, \varphi_g \in\mathcal{D}(\mathbb{R}^d)$ as per hypothesis. 
On account of Theorem \ref{th:inclusions_positive_regularity} and of Theorem \ref{th:inclusions_negative_regularity}, it holds
\begin{gather*}
\varphi_ff \in \mathcal{C}_{c}^\alpha(\mathbb{R}^d)\;\Rightarrow\;\varphi_ff\in H^{s_1}(\mathbb{R}^d)\,,\forall\,s_1<\alpha\,,\\
\varphi_g g\in \mathcal{C}_{c}^\beta(\mathbb{R}^d)\;\Rightarrow \varphi_g g \in B^\gamma_{p,\infty,c}(\mathbb{R}^d)\,,\,\forall\,\gamma \in [\beta, \beta^\ast_{\varphi_g g}) \text{ for some } p\in[2,\infty] \Rightarrow \varphi_gg\in H^{s_2}(\mathbb{R}^d)\,,\forall\,s_2<\beta^\ast_{\varphi_g g}\,.
\end{gather*}
Definition \ref{Def: Sobolev WF} entails that $\mathrm{WF}^{s_1}(\varphi_ff)=\emptyset$ for $s_1=\alpha-\varepsilon$ and $\mathrm{WF}^{s_2}(\varphi_gg)=\emptyset$ for $s_2=\beta^\ast_{\varphi_g g}-\delta$ for any $\varepsilon,\delta\in(0,1)$. 
We observe that, on account of Equation \eqref{Eq: Novel Young condition}, it holds $s_1+s_2\geq0$ and as a consequence one can apply Proposition \ref{Prop: Sobolev product of distributions}. This concludes the proof.
\end{proof}

In the following examples, we shall focus on specific situations in which Theorem \ref{Thm: classical Young} does not directly apply but Theorem \ref{Thm: microlocal Young} does.
\begin{example}
Let $f\in\mathcal{C}^d(\mathbb{R}^d)$ be such that $\mathrm{singsupp}(f)=\{0\}$ and let $\delta_0\in\mathcal{C}^{-d}(\mathbb{R}^d)$ be the Dirac delta distribution centred at the origin.
While Theorem \ref{Thm: classical Young} does not apply since $\alpha+\beta=0$, the hypothesis of Theorem \ref{Thm: microlocal Young} are met since Equation \eqref{Eq: Novel Young condition} reduces to

\begin{align*}
d - \frac{d}{2} = \frac{d}{2}>0\,.
\end{align*}  
As a consequence of Theorem \ref{Thm: microlocal Young}, the product $f\cdot\delta_0\in\mathcal{D}^\prime(\mathbb{R}^d)$ exists. 
Moreover, it is defined by $(f\delta_0)(\varphi)=\delta_0(f\varphi)$ for any $\varphi\in\mathcal{D}(\mathbb{R}^d)$, yielding $f\delta_0=f(0)\delta_0$.
Observe that we could have also picked $f\in\mathcal{C}^{\frac{d}{2}+\varepsilon}(\mathbb{R}^d)$ for any $\varepsilon>0$ without hindering the existence of the pointwise product $f\cdot\delta_0$. Yet, in such a scenario, with the notations of Theorem \ref{Thm: classical Young}, $\alpha+\beta<0$.
\end{example}

\begin{example}
Let $f \in \mathcal{C}^d(\mathbb{R}^d)$ with $d > 2$ be such that $\text{singsupp}(f)=\{0,y\}$ with $y \neq 0$. 
Moreover let $g \in \mathcal{C}^{-d-1}(\mathbb{R}^d)$ be $g = \delta_0 + \delta_y^{(1)}$, where $\delta_0$ is the Dirac delta centered at the origin while $\delta_y^{(1)}$ is the distributional derivative of the Dirac delta centered at $0\neq y \in \mathbb{R}^d$. Since $\alpha + \beta = d -d-1=-1$, then the classical Young theorem, \textit{cf.} Theorem \ref{Thm: classical Young} cannot be applied and we cannot multiply $f$ and $g$. 
Nevertheless, the hypotheses of Theorem \ref{Thm: microlocal Young} are met. 
If we consider the origin, there exist two test functions $\varphi_f\in \mathcal{D}$ with $\varphi_f(0)\neq 0$ and $\varphi_g \in \mathcal{D}$ with $\varphi_g(0)\neq 0$ and $y \not \in \text{supp}(\varphi_g)$ such that the condition on the regularities, represented by Equation \eqref{Eq: Novel Young condition} is fulfilled. 
It descends
\[
\alpha + \beta^\ast_{\varphi_g g} = d -d/2 = d/2 > 0\,.
\]
Concerning the point $y$, we can localize $f$ and $g$ by two test functions $\psi_f\in \mathcal{D}$ with $\psi_f(y)\neq 0$ and $\psi_g \in \mathcal{D}$ with $\psi_g(y)\neq 0$ and $0 \not \in \text{supp}(\psi_g)$ in such a way that
\[
\alpha + \beta^\ast_{\psi_g g} = d - d/2 - 1 = d/2 - 1 > 0\,,
\]
where we used the hypothesis $d > 2$ and the fact that $\beta^\ast_{\delta_y^{(1)}} = -d/2 - 1$. 
Thus there exists $f \cdot g \in \mathcal{D}^\prime(\mathbb{R}^d)$ as per Theorem \ref{Thm: microlocal Young}.
\end{example}

The hypothesis of Theorem \ref{Thm: microlocal Young} are quite general but, as a matter of fact, if we make further assumptions concerning the regularity of the distributions we wish to multiply, then we can draw more conclusions on the regularity of their product and on the continuity of such operation.

\begin{theorem}\label{Thm: microlocal Young II}
Let $f \in \mathcal{C}_{\mathrm{loc}}^\alpha(\mathbb{R}^d)$ and $g \in \mathcal{C}_{\mathrm{loc}}^\beta(\mathbb{R}^d)$ with $\alpha > 0$ and $\beta < 0$. If $\text{singsupp}(f) \cap \text{singsupp}(g) \neq \emptyset$ and 
\begin{align}\label{Eq: novel Young condition II}
\alpha + \beta^\ast_g > 0\,,
\end{align} 
there exists the product $f \cdot g \in \mathcal{C}_{\mathrm{loc}}^\beta(\mathbb{R}^d)$. Moreover for any compact set $K \subset \mathbb{R}^d$, one has
\begin{align}\label{Eq: continuity of microlocal young product}
\|f\cdot g\|_{\mathcal{C}^{\beta}(K)}\lesssim\|f\|_{\mathcal{C}^{\alpha}(\overline{K}_1)}\|g\|_{\mathcal{C}^{\beta}(\overline{K}_1)}\,,
\end{align}
where $\overline{K}_1$ denotes the $1$-enlargement of $K$.
\end{theorem}
\begin{proof}
Observe that the hypotheses of Theorem \ref{Thm: microlocal Young} are met and thus $f\cdot g\in\mathcal{D}^\prime(\mathbb{R}^d)$ and, in addition, $(f\cdot g)(\psi)=g(f\psi)$ for any $\psi\in\mathcal{D}(\mathbb{R}^d)$.
Using Definition \ref{Def: C-alpha-loc}, for any but fixed compact set $K\subset\mathbb{R}^d$, it follows that
\begin{align*}
f(y)=P_x(y)+R(x,y)\,,\qquad\forall\,x,y\in K\,,
\end{align*} 
where $P_x(\cdot)$ is the Taylor polynomial of $f$ centered at $x$ of order $\lfloor\alpha\rfloor$, while $R$ is the remainder which satisfies $\lvert R(x,y)\rvert\leq\|f\|_{\mathcal{C}^\alpha(K)}\lvert x-y\rvert^\alpha$. In addition $R$ can be written as
\begin{equation}\label{eq:remainder}
R(x,y) = \sum_{\lvert k \rvert = \lfloor \alpha \rfloor} h_k(y) (y-x)^k\,,
\end{equation}
where $h_k$ is such that
\begin{align*}
\lim_{y \to x} \frac{\lvert h_k(y) \rvert}{\lvert y - x \rvert^{\alpha - \lfloor \alpha \rfloor}} = C_f \leq \|f\|_{\mathcal{C}^\alpha(K)}\,.
\end{align*}
As a consequence, for any $x\in K$, $\psi\in\mathcal{D}(B(0,1))$, $\lambda\in(0,1]$ and denoting with $\overline{K}_1$ the $1$-enlargement of $K$, we have
\begin{align*}
\lvert(fg)(\psi^\lambda_x)\rvert
&=\lvert g(f\psi^\lambda_x)\rvert
\leq\underbrace{\lvert g(P_x(\cdot)\psi^\lambda_x)\rvert}_{\lvert A\rvert}+\underbrace{\lvert g(R(x,\cdot)\psi^\lambda_x)\rvert}_{\lvert B \rvert}\,.
\end{align*}
Focusing on $\lvert A\rvert$, on account of the triangular inequality and being $g$ a linear map we obtain 
\begin{align*}
\lvert A\rvert\leq\sum_{\lvert k\rvert<\alpha}\frac{\lvert\partial^kf(x)\rvert}{k!}\lvert g((\cdot-x)^k\psi^\lambda_x)\rvert\,.
\end{align*}
Setting $\eta(y)\vcentcolon=y^k\psi(y)$, we have
\begin{align*}
(y-x)^k \psi^\lambda_x(y) = \lambda^{\lvert k \rvert} \eta^\lambda_x(y)\,.
\end{align*}
It follows, 
\begin{equation}
\lvert g((\cdot-x)^k \psi^\lambda_x) \rvert = \lambda^{\lvert k \rvert} \lvert g(\eta^\lambda_x)\rvert \lesssim \|g\|_{\mathcal{C}^\beta(\bar{K}_1)} \lambda^{\beta+\rvert k \rvert}\,,
\end{equation}
and then
\begin{align*}
\lvert A\rvert\lesssim\|g\|_{\mathcal{C}^\beta(\bar{K}_1)}\sum_{\lvert k\rvert<\alpha} \frac{\|\partial^kf\|_{L^\infty(K)}}{k!}\lambda^{\beta+\rvert k\rvert}\lesssim\|f\|_{\mathcal{C}^\alpha(\bar{K}_1)}\|g\|_{\mathcal{C}^\beta(\bar{K}_1)}\lambda^\beta\,,
\end{align*}
where we used the characterization of $g$ as per Definition \ref{Def: C-alpha-loc-negativi} and $\|\partial^kf\|_{L^\infty(K)}\lesssim\|f\|_{\mathcal{C}^\alpha(\bar{K}_1)}$. 
Focusing on $\lvert B\rvert$, we have
\begin{align*}
\lvert B\rvert\leq \sum_{\lvert k\rvert=\lfloor\alpha\rfloor}\lvert g(h_k(\cdot)(\cdot-x)^k \psi_x^\lambda)\rvert\,,
\end{align*}
and, setting again $\eta(y)\vcentcolon=y^k\psi(y)$, we have
\begin{align*}
\lvert B\rvert\leq\sum_{\lvert k\rvert=\lfloor\alpha\rfloor}\lambda^{\lvert k\rvert}\lvert g(h_k(\cdot)\eta_x^\lambda)\rvert\,.
\end{align*}
In the limit $\lambda\to0^+$, $h_k(y)\approx C_f\sum_{\lvert m\rvert=\alpha-\lfloor\alpha\rfloor}(y-x)^{m}$ for $y$ sufficiently close to $x$. Hence
\begin{align*}
\lvert B \rvert \lesssim C_f \sum_{\lvert k \rvert = \lfloor \alpha \rfloor} \sum_{\lvert m \rvert = \alpha -\lfloor \alpha \rfloor}\lambda^{\lvert k \rvert}\lvert g((\cdot-x)^{m}\eta_x^\lambda)\rvert\lesssim C_f\|g\|_{\mathcal{C}^\beta(\bar{K}_1)}\lambda^\beta\leq\|f\|_{\mathcal{C}^\alpha(\bar{K}_1)}\|g\|_{\mathcal{C}^\beta(\bar{K}_1)}\lambda^\beta\,.
\end{align*}
Together with the estimates for $|A|$, this entails $f\cdot g\in\mathcal{C}^{\beta}_{\mathrm{loc}}(\mathbb{R}^d)$ as well as Equation \eqref{Eq: continuity of microlocal young product}. 
\end{proof}

\begin{example}
Let $g\colon \mathbb{R}^d\to \mathbb{R}$ be $g(x) = \lvert x \rvert^\beta$ with $\beta < -d/2$ and let $f \in \mathcal{C}^\alpha(\mathbb{R}^d)$ with $\alpha > 0$ be such that $0 \in \text{singsupp}(f)$. Then there exists the product $f \cdot g$ if $\alpha + \beta^\ast_g > 0$, {\it i.e.} if $\alpha + \beta > -d/2$.
\end{example}

\begin{remark}
We observe that in the hypotheses of Theorem \ref{Thm: microlocal Young II}, if the singular supports of the underlying distributions $f$ and $g$ are disjoint, we can conclude that $f\cdot g\in\mathcal{C}^{\beta}_{\mathrm{loc}}(\mathbb{R}^d)$ since locally this product is that between a smooth function and a distribution.
\end{remark}

We observe that Theorem \ref{Thm: microlocal Young} and Theorem \ref{Thm: microlocal Young II}, particularly Equations \eqref{Eq: Novel Young condition} and \eqref{Eq: novel Young condition II}, entail that, in order to multiply $f\in\mathcal{C}^\alpha_{\textrm{loc}}(\mathbb{R}^d)$ and $g\in\mathcal{C}^\beta_{\textrm{loc}}(\mathbb{R}^d)$, at least one among the parameters $\alpha$ and $\beta$ has to be strictly positive.
In the remaining part of this section, we shall scrutinize the cases where Equation \eqref{Eq: Novel Young condition} is not satisfied and, hence the product between $f$ and $g$ \textit{a priori} does not exist.
Nonetheless, using the notion of scaling degree, see Appendix \ref{App: sd}, it is possible to establish sufficient conditions allowing to bypass this obstruction. 
As first step, we consider $f\in\mathcal{C}^\alpha_{\mathrm{loc}}(\mathbb{R}^d)$ and $g\in\mathcal{C}^\beta_{\mathrm{loc}}(\mathbb{R}^d)$ having disjoint singular supports, while $\alpha,\beta<0$. 
The following result holds true.

\begin{lemma}\label{Lem: disjoint singsupps}
Let $\alpha,\beta<0$ and let $f\in\mathcal{C}^\alpha_{\mathrm{loc}}(\mathbb{R}^d)$ and $g\in\mathcal{C}^\beta_{\mathrm{loc}}(\mathbb{R}^d)$ be such that $\operatorname{singsupp}(f)\cap\operatorname{singsupp}(g)=\emptyset$. 
Then there exists a canonical product $f\cdot g\in \mathcal{C}^{\alpha+\beta}_{\mathrm{loc}}(\mathbb{R}^d)$.
\end{lemma}
\begin{proof}
Since the singular supports are disjoint, there exists $f\cdot g\in\mathcal{D}^\prime(\mathbb{R}^d)$. 
Moreover, for any $x\in\mathbb{R}^d$
\begin{align*}
\mathrm{sd}_x(f\cdot g)\leq\mathrm{sd}_x(f)+\mathrm{sd}_x(g)\leq-\alpha-\beta\,,
\end{align*}
where $sd_x$ denotes the scaling degree at $x\in\mathbb{R}^d$, see Definition \ref{Def: scaling degree} and where we exploited Remark \ref{Rmk: sd product estimate}.
Thus, on account of Proposition \ref{Prop: link between sd and Holder}, $f\cdot g$ belongs to $\mathcal{C}^{-\mathrm{sd}_x(f\cdot g)}$ around $\{x\}$, \textit{i.e.}, $f\cdot g\in\mathcal{C}^{\alpha+\beta}_{\mathrm{loc}}(\mathbb{R}^d)$. 
\end{proof}

\begin{remark}\label{Rem: worsening of regularity}
	Observe that in Lemma \ref{Lem: disjoint singsupps} and in the following statements, we can establish only that $f\cdot g\in\mathcal{C}^{\alpha+\beta}_{\mathrm{loc}}(\mathbb{R}^d)$. This is a space containing more singular distributions in comparison to the results of Theorem \ref{Thm: microlocal Young} and Theorem \ref{Thm: microlocal Young II} in which we establish that $f\cdot g\in\mathcal{C}^{\beta}_{\mathrm{loc}}(\mathbb{R}^d)$.
\end{remark}

Subsequently, suppose $f\in\mathcal{C}^\alpha_{\mathrm{loc}}(\mathbb{R}^d)$ and $g\in\mathcal{C}^\beta_{\mathrm{loc}}(\mathbb{R}^d)$ and assume that there exists $x\in\mathbb{R}^d$ such that $\{x\}=\operatorname{singsupp}(f)\cap\operatorname{singsupp}(g)$. If Equation \eqref{Eq: novel Young condition II} does not hold true, our strategy to define the product between $f$ and $g$ consists of applying Lemma \ref{Lem: disjoint singsupps} on $\mathbb{R}^d\setminus\{x\}$, using subsequently an extension procedure as per Appendix \ref{App: sd}. This is the spirit at the heart of the following theorem.

\begin{theorem}\label{Thm: microlocal Young with sd}
Let $f\in\mathcal{C}^\alpha_{\mathrm{loc}}(\mathbb{R}^d)$ and $g\in\mathcal{C}^\beta_{\mathrm{loc}}(\mathbb{R}^d)$, with $\alpha,\beta\in\mathbb{R}$. Assume that $\operatorname{singsupp}(f)\cap\operatorname{singsupp}(g)=\{x\}$ with $x\in\mathbb{R}^d$.
Then
\begin{enumerate}[(i)]
\item If $\alpha>0$ and $\beta\in(-d,0)$ are such that 
\begin{align*}
\alpha+\beta^*_{g}\leq0\,,
\end{align*}
there exists a unique extension $\widetilde{f\cdot g}\in\mathcal{C}^\beta_{\mathrm{loc}}(\mathbb{R}^d)$  of $f\cdot g\in\mathcal{C}^\beta_{\mathrm{loc}}(\mathbb{R}^d\setminus\{x\})$ which preserves the scaling degree;
\item If $\alpha<0$ and $\beta<0$, there exists an extension $\widetilde{f\cdot g}\in \mathcal{C}^{\alpha+\beta}_{\mathrm{loc}}(\mathbb{R}^d)$ of $f\cdot g\in\mathcal{C}^{\alpha+\beta}_{\mathrm{loc}}(\mathbb{R}^d\setminus\{x\})$ which preserves the scaling degree. 
Moreover, if $\alpha+\beta>-d$, such extension is unique.
\end{enumerate}
\end{theorem}

\begin{proof}
$\textbf{\textit{(i)}:}$ Starting from the case $\alpha>0$ and $\beta\in(-d,0]$, Theorem \ref{Thm: microlocal Young II} entails that there exists $f\cdot g\in\mathcal{C}^\beta_{\mathrm{loc}}(\mathbb{R}^d\setminus\{x\})$ where, with a slight abuse of notation we are indicating with the same symbols the restrictions of $f$ and $g$ to $\mathbb{R}^d\setminus\{x\}$.
In addition observe that, on account of Remark \ref{Rmk: sd product estimate} and of Equation \eqref{Eq: Holder characterization of negative regularity},
\begin{align*}
\mathrm{sd}_x(f\cdot g)\leq\mathrm{sd}_x(f)+\mathrm{sd}_x(g)\leq-\beta<d\,,
\end{align*}
where we exploited $\mathrm{sd}_x(f)\leq0$, which is a consequence of the continuity of $f$ as a function, being $\alpha>0$.
As a consequence, we can apply Theorem \ref{Thm: extension with scaling degree} which entails the existence of a unique extension $\widetilde{f\cdot g}\in\mathcal{D}^\prime(\mathbb{R}^d)$ of $f\cdot g\in\mathcal{C}^\beta_{\mathrm{loc}}(\mathbb{R}^d\setminus\{x\})$ which preserves the scaling degree. 
As a consequence, in a neighbourhood of $x$ we have $\widetilde{f\cdot g}\in\mathcal{C}^{-\mathrm{sd}_x(f\cdot g)}\subseteq\mathcal{C}^\beta$. This entails $\widetilde{f\cdot g}\in\mathcal{C}^{\beta}_{\mathrm{loc}}(\mathbb{R}^d)$.

$\textbf{\textit{(ii)}:}$ In this scenario, on account of Lemma \ref{Lem: disjoint singsupps} we can conclude that $f\cdot g\in\mathcal{C}^{\alpha+\beta}_{\mathrm{loc}}(\mathbb{R}^d\setminus\{x\})$. 
Moreover, once more on account of Remark \ref{Rmk: sd product estimate}, $\mathrm{sd}_x(f\cdot g)\leq-\alpha-\beta$. 
Similarly to the previous case, we apply Theorem \ref{Thm: extension with scaling degree} to conclude the existence of at least one extension $\widetilde{f\cdot g}\in\mathcal{D}^\prime(\mathbb{R}^d)$ of $f\cdot g\in\mathcal{C}^{\alpha+\beta}_{\mathrm{loc}}(\mathbb{R}^d\setminus\{x\})$ and to conclude that $\widetilde{f\cdot g}\in\mathcal{C}^{\alpha+\beta}_{\mathrm{loc}}(\mathbb{R}^d)$.
We underline that if $\alpha+\beta>-d$, Theorem \ref{Thm: extension with scaling degree} guarantees also uniqueness of such an extension. 
\end{proof}

\begin{example}
As an example, we consider the case of the distribution generated by the function $u(x)=|x|^{-\frac{1}{2}}$, with $x\in\mathbb{R}\setminus\{0\}$ and its product with itself.
First of all, we notice that $u\in\mathcal{C}^{-\frac{1}{2}}(\mathbb{R})$ and thus we fulfill the hypotheses of item $\textit{(ii)}$ of Theorem \ref{Thm: microlocal Young with sd} since $\operatorname{singsupp}(u)=\{0\}$.
We observe that $u\cdot u\in\mathcal{D}^\prime(\mathbb{R}\setminus\{0\})$, with $u\cdot u$ being the distribution generated by the function $u^2(x)=|x|^{-1}$, but  $u\cdot u\notin\mathcal{D}^\prime(\mathbb{R})$.
\noindent Focusing on the scaling degree at $0$ of $u\cdot u$, applying Definition \ref{Def: scaling degree}, $\mathrm{sd}_0(u\cdot u)\leq2\mathrm{sd}_0(u)=1$. 
Since $d=1$, it follows by Theorem \ref{Thm: microlocal Young with sd} that an extension exists but it is not unique. 
In particular, dwelling more into the details of Theorem \ref{Thm: extension with scaling degree}, one can conclude that there exists a one-parameter family $\{\widetilde{u\cdot u}_C\}_{C\in\mathbb{R}}\subset\mathcal{C}^{-1}(\mathbb{R})$ of extensions preserving the scaling degree.
This is the family of distributions generated by 
\begin{align*}
(\widetilde{u\cdot u})_C(x)\vcentcolon=\frac{d}{dx}\log|x|+C\delta_0(x)\,,
\end{align*}
where formally $\delta_0(x)$ denotes the integral kernel of the Dirac delta distribution centred at the origin.
\end{example}

\begin{remark}
In the previous theorem we assumed the intersection of the singular supports of $f$ and $g$ to be a point $x$. 
Nonetheless, the above argument can be generalized straightforwardly to more complex scenarios. 
As an example, if the intersection of the singular supports is given by a countable number of separated points, the same argument can be applied locally around any of these points. 
Subsequently, one can reconstruct globally the distribution $\widetilde{f\cdot g}$, of regularity $\mathcal{C}^{\alpha+\beta}_{\mathrm{loc}}(\mathbb{R}^d)$, through a partition of unity argument.  

Similarly, this argument can be extended to cases in which the intersection of the singular supports is a sub-manifold $\Sigma\subset\mathbb{R}^d$ of codimension at least $1$.
In such a case, one introduces an open covering $\{U_{j}\}_{j\in\mathbb{N}}$ of this submanifold. 
In view of the previous results, for any $j\in\mathbb{N}$, $(\varphi_j f)\cdot(\varphi_j g)\in\mathcal{D}^\prime(U_j\setminus(U_j\cap\Sigma))$ for a suitable test-function $\varphi_j\in\mathcal{D}(U_j)$ acting as a localization. 
Through the notion of scaling degree with respect to a submanifold \cite{BF00}, which is a generalization of the case discussed here, one can implement an extension procedure similar to the one of Theorem \ref{Thm: microlocal Young with sd} -- yielding a distribution on $U_j$. 
One can conclude the construction using a partition of unity argument.
\end{remark}

\begin{remark}
We stress that there exists notable instances in which the hypotheses of the previous theorems do not hold true. 
As an example it is not possible to multiply a Gaussian white noise $\xi\in\mathcal{C}^{-\frac{d}{2}-\varepsilon}(\mathbb{R}^d)$ (almost surely) with itself. This because the singular support of $\xi$ coincides almost surely with the whole $\mathbb{R}^d$, since formally the white noise is the distributional derivative of a Brownian motion, which admits no smooth restriction to any neighbourhood $U$ of any point $x\in\mathbb{R}^d$.
\end{remark}

\subsection{Application to Coherent Germs of Distributions}\label{Sec: Application to germs}

In this last part of the section we shall discuss an application of the results of the previous section in the context of coherent germs of distributions \cite{CZ20, RS21}, which provides an alternative viewpoint to the renown reconstruction theorem due to Hairer \cite{Hai14}.
In particular we shall discuss how this result provides a rationale for identifying a canonical reconstructed distribution of a suitable class of coherent germs of distributions which otherwise would be non-canonical. 

As a premise, we briefly recall the definitions of germ of distributions and of coherence, together with the statement of the reconstruction theorem as per \cite{CZ20}.

\begin{definition}
A family $(F_x)_{x\in \mathbb{R}^d}$ of distributions $F_x\in\mathcal{D}^\prime(\mathbb{R}^d)$ is called \textit{germ of distributions} if for any $\psi\in\mathcal{D}(\mathbb{R}^d)$ the map $\mathbb{R}^d\ni x\mapsto F_x(\psi)\in\mathbb{R}$ is measurable.
\end{definition}

\begin{definition}
Let $\gamma\in\mathbb{R}$. A germ of distributions $F=(F_x)_{x\in\mathbb{R}^d}$ is $\gamma$-coherent if, for any compact set $K\subset\mathbb{R}^d$, $\exists\alpha_K \leq\min\{\gamma,0\}$ such that for any $r >-\alpha_K$ 
\begin{align}\label{Eq: coherence}
|(F_x-F_y)(\varphi^\lambda_y)| \lesssim \lambda^{\alpha_K} (|x-y|+\lambda)^{\gamma-\alpha_K}\,,
\end{align}
uniformly for $x,y\in K$, $\lambda\in(0,1]$ and $\varphi\in\mathcal{B}_{r}\vcentcolon=\{\varphi\in\mathcal{D}(B(0,1))\;|\;\|\varphi\|_{C^r}\leq1\}$. If $\bm{\alpha}=(\alpha_K)$ is the family of exponents in Equation \eqref{Eq: coherence}, the germ $F$ is called $(\bm{\alpha},\gamma)$-coherent.
\end{definition}

\noindent Finally, we recall the statement of the reconstruction theorem.

\begin{theorem}[Reconstruction Theorem]\label{Thm: reconstruction theorem}
Let $F=(F_x)_{x\in\mathbb{R}^d}$ be a $(\bm{\alpha},\gamma)$-coherent germ with $\gamma\in\mathbb{R}$. Then  $\exists\mathcal{R}F\in\mathcal{D}^\prime(\mathbb{R}^d)$, dubbed \textit{reconstruction} of $F$, such that for any compact set $K\subset\mathbb{R}^d$ 
\begin{align}\label{Eq: reconstruction bound}
\lvert(\mathcal{R}F-F_x)(\varphi_x^\lambda)\rvert\lesssim 
\begin{cases}
\lambda^\gamma\,,\qquad&\text{if}\;\gamma \neq 0\,, \\
1+\lvert\log\lambda\rvert\,,\qquad&\text{if}\;\gamma=0\,,
\end{cases}
\end{align}
uniformly for $x\in K$, $\lambda\in(0,1]$, $\varphi\in\mathcal{D}(\mathbb{R}^d)$. Furthermore, if $\gamma>0$, the  distribution $\mathcal{R}F$ is unique.
\end{theorem}

Within this framework, we use the previous results to prove the following statement which improves on the derivation of Young's product theorem in \cite{CZ20}.

\begin{proposition}\label{Prop: application_to_germs}
	Let $\alpha>0$, $\beta<0$ and let $f\in\mathcal{C}_{\mathrm{loc}}^\alpha(\mathbb{R}^d)$ and $g\in\mathcal{C}_{\mathrm{loc}}^\beta(\mathbb{R}^d)$ be such that they abide to the hypotheses of Theorem \ref{Thm: microlocal Young II}. Then, calling
	\begin{align*}
	P_x(y)=\sum_{k<\alpha}\frac{\partial^kf(x)}{k!}(y-x)^k\,,
	\end{align*}	
	the Taylor polynomial of $f$, $F=(F_x)_{x\in\mathbb{R}^d}$ where
	\begin{align}\label{Eq: germ of the product}
		F_x(\bullet)\vcentcolon=(P_x\cdot g)(\bullet)=\biggl(\sum_{\lvert k \rvert < \alpha} \frac{\partial^k f(x)}{k!} (\bullet- x)^k g(\bullet)\biggr)_{x\in\mathbb{R}^d}\,.
	\end{align}
is a $(\beta,\alpha+\beta)$-coherent germ, whose reconstruction as per Theorem \ref{Thm: reconstruction theorem} is $\mathcal{R}F=f\cdot g$.
\end{proposition}

\begin{proof}
	Following \cite[Prop 14.4]{CZ20} we can infer that $F$ is a $(\beta,\alpha+\beta)$-coherent germ. Since the hypotheses of Theorem \ref{Thm: microlocal Young II} are met, then the product $f\cdot g$ exists and it is uniquely defined as $\Delta^*(f\otimes g)$ where $\Delta$ is the diagonal map. Finally, to prove that $\mathcal{R}(F)=f\cdot g$, we follow slavishly the proof of Theorem \ref{Thm: microlocal Young II} and we observe that, in view of Equation \eqref{Eq: reconstruction bound},  for any compact set $K\subseteq\mathbb{R}^d$, in the limit $\lambda\to0^+$,
	\begin{align*}
		\lvert(f\cdot g-P_x\cdot g&)(\varphi^\lambda_x)\rvert=\lvert g((f-P_x)\varphi^\lambda_x) \rvert=\lvert g(R(x,\cdot)\varphi^\lambda_x)\rvert \nonumber \\
		&\leq \sum_{\lvert k \rvert = \lfloor \alpha \rfloor} \lvert g(h_k(\cdot)(\cdot-x)^k \varphi^\lambda_x)\rvert \overset{\tilde{\varphi}(y)=y^k \varphi(y)}{=} \sum_{\lvert k \rvert = \lfloor \alpha \rfloor} \lambda^{\lvert k \rvert} \lvert g(h_k(\cdot) \tilde{\varphi}^\lambda_x)\rvert = \lambda^{\lfloor \alpha \rfloor}\sum_{\lvert k \rvert = \lfloor \alpha \rfloor}  \lvert g(h_k(\cdot) \tilde{\varphi}^\lambda_x)\rvert \nonumber \\
&=\lambda^{\lfloor \alpha \rfloor}\sum_{\lvert k \rvert = \lfloor \alpha \rfloor} \sum_{\lvert m \rvert = \alpha - \lfloor \alpha \rfloor}  \lvert g(C_f (\cdot- x)^m \tilde{\varphi}^\lambda_x)\rvert \overset{\eta(y)=y^m\tilde{\varphi}}{\lesssim} \lambda^{\alpha} \lvert g(\eta^\lambda_x) \rvert \lesssim \lambda^{\alpha+\beta}\,,
	\end{align*}
	uniformly for $x\in K$, $\varphi\in\mathcal{D}(B(0,1))$. Here we exploited the properties of $\mathcal{C}^\alpha_{\mathrm{loc}}(\mathbb{R}^d)$, $\alpha>0$ and of $\mathcal{C}^\beta_{\mathrm{loc}}(\mathbb{R}^d)$, $\beta<0$.
\end{proof}

\begin{remark}
Observe that, if $\alpha+\beta>0$ in Proposition \ref{Prop: application_to_germs}, we have reproduced the results of \cite[Thm. 14.1]{CZ20}. At the same time, if $\alpha+\beta\leq 0$, the reconstruction of the germ of distributions of Equation \eqref{Eq: germ of the product} as per \cite{CZ20} is not unique and this reverberates in a lack of a canonical choice of the product between $f$ and $g$. 

From this point of view, Theorem \ref{Thm: microlocal Young II} provides a rationale to identify, among all the possible reconstructions, a \textit{canonical} reconstruction of the germ of distributions of Equation \eqref{Eq: germ of the product} represented by the product $f\cdot g$ in cases where the classical Young product theorem fails to do so. 
\end{remark}

\begin{remark}
Observe that, in the case $\alpha+\beta=0$ the reconstruction theorem as per \cite{CZ20, Hai14} entails that 
\begin{align*}
\lvert(f\cdot g-P_x\cdot g)(\varphi^\lambda_x)\rvert\lesssim (1+|\log\lambda|)\,,
\end{align*}
whereas Proposition \ref{Prop: application_to_germs} yields
\begin{align*}
\lvert(f\cdot g-P_x\cdot g)(\varphi^\lambda_x)\rvert\lesssim1\,.
\end{align*}
The two statements are not in contradiction one with each other and, actually, that of Proposition \ref{Prop: application_to_germs} improves on the standard result from the reconstruction theorem, though its application is limited to the product of those distributions abiding to Theorem \ref{Thm: microlocal Young II}.
\end{remark}

\appendix
\section{Scaling Degree}\label{App: sd}
In this appendix we shall briefly recall the definition and some results related to the theory of the scaling degree of distributions, without any attempt of completeness. 
For further details, we refer to \cite{BF00} and \cite[App. B]{DDRZ20}.

\begin{definition}\label{Def: scaling degree}
	Let $U\subseteq \mathbb{R}^d$ be a conic open subset of $\mathbb{R}^d$ and for any $x_0\in \mathbb{R}^d$, let $U_{x_0}\vcentcolon=U+x_0$.
	Given $f\in\mathcal{D}(U)$ and $\lambda\in(0,1)$, we denote $f^\lambda_{x_0}\vcentcolon=\lambda^{-d}f(\lambda^{-1}(x-x_0))\in\mathcal{D}(U_{x_0})$.
	By duality, given $t\in\mathcal{D}'(U_{x_0})$ we define $t^\lambda_{x_0}\in\mathcal{D}'(U)$ as $t^\lambda_{x_0}(f)\vcentcolon=t(f^\lambda_{x_0})$ for any $f\in\mathcal{D}(U)$. 
	The scaling degree of $t$ at $x_0$ is 
	\begin{align}\label{Eq: scaling degree at x0}
	\operatorname{sd}_{x_0}(t)
	\vcentcolon=\inf\bigg\lbrace\omega\in\mathbb{R}\,|\,\lim_{\lambda\to 0^+}\lambda^\omega t^\lambda_{x_0}=0\bigg\rbrace\,.
	\end{align}
\end{definition}

\begin{example}\label{Ex: sd of delta on a point}
	As an example, we consider the Dirac delta distribution centered at a point $x\in\mathbb{R}^d$, $\delta_{x}\in\mathcal{D}'(\mathbb{R}^d)$.
	A direct computation shows that $\delta_x^\lambda=\lambda^{-d}\delta_x$ so that $\operatorname{sd}_x(\delta_x)=d$.
\end{example}

\begin{remark}\label{Rmk: sd product estimate}
We observe that the scaling degree is additive with respect to tensor product, \textit{i.e.}, \cite[Lem. 5.1]{BF00}
	\begin{align}\label{Eq: sd of tensor product}
		\operatorname{sd}_{(x_1,x_2)}(T_1\otimes T_2)
		=\operatorname{sd}_{x_1}(T_1)
		+\operatorname{sd}_{x_2}(T_2)\,.
	\end{align}
	Moreover \cite[Lem. 6.6]{BF00} entails (in particular) that if $T_1,T_2\in\mathcal{D}'(\mathbb{R}^d)$ are such that their product exists, see \cite[Thm. 8.2.10]{Hor03}, then $T_1T_2\in\mathcal{D}'(\mathbb{R}^d)$ satisfies
	\begin{align}\label{Eq: sd of product}
		\operatorname{sd}_x(T_1T_2)
		\leq\operatorname{sd}_x(T_1)
		+\operatorname{sd}_x(T_2)\,.
	\end{align}
\end{remark}

The following theorem \cite[Thm. 5.2-5.3]{BF00} connects the notion of scaling degree at a point with an extension procedure for distributions, of which we recall the notion.

\begin{definition}
Let $x\in\mathbb{R}^d$ and $t\in\mathcal{D}'(\mathbb{R}^d\setminus\{x\})$. We say that $\widetilde{t}\in\mathcal{D}^\prime(\mathbb{R}^d)$ is an extension of $t$ and we denote it with $t\subseteq\widetilde{t}$ if for any $\psi\in\mathcal{D}(\mathbb{R}^d\setminus\{x\})$ it holds $\widetilde{t}(\psi)=t(\psi)$.
\end{definition}

\begin{theorem}\label{Thm: extension with scaling degree}
	Let $x\in\mathbb{R}^d$, $t\in\mathcal{D}'(\mathbb{R}^d\setminus\{x\})$ and set $\rho:=\operatorname{sd}_x(t)-d$.
	Then
	\begin{enumerate}
		\item
		if $\rho<0$, there exists a unique $\widetilde{t}\in\mathcal{D}'(\mathbb{R}^d)$ such that $t\subseteq\widetilde{t}$ and $\operatorname{sd}_x(\widetilde{t})=\operatorname{sd}_x(t)$.
		\item
		if $\rho\geq 0$, any distribution $\widetilde{t}\in\mathcal{D}'(\mathbb{R}^d)$ such that $t\subseteq\widetilde{t}$ and $\operatorname{sd}_x(\widetilde{t})=\operatorname{sd}_x(t)$ is of the form
		\begin{align*}
			\widetilde{t}
			=t\circ W_\rho
			+\sum_{|\alpha|\leq\rho}a_\alpha\partial^\alpha\delta_x\,,
		\end{align*}
		where $\{a_\alpha\}_\alpha\subset\mathbb{C}$ while $W_\rho\colon\mathcal{D}(\mathbb{R}^d)\to\mathcal{D}(\mathbb{R}^d)$ is defined by
		\begin{align*}
			W_\rho f
			:=f
			-\sum_{|\alpha|\leq \rho}\frac{1}{\alpha!}(\partial^\alpha f)(x)\psi_\alpha\,,
		\end{align*}
		with $\{\psi_\alpha\}_\alpha\subseteq\mathcal{D}(\mathbb{R}^d)$ any family test-functions satisfying $\partial^\beta\psi_\alpha(x)=\delta^\beta_\alpha$.
		\item
		if $\rho=+\infty$, no extensions $\widetilde{t}\in\mathcal{D}'(\mathbb{R}^d)$ of $t$ exist.
	\end{enumerate}
\end{theorem}

Finally, we state a result connecting the notion of scaling degree and the space of H\"older functions of negative regularity.

\begin{proposition}\label{Prop: link between sd and Holder}
Let $u\in\mathcal{D}^\prime(\mathbb{R}^d)$ such that there exists $\beta\vcentcolon=\sup_{x\in\mathbb{R}^d}\mathrm{sd}_x(u)>0$. 
Then $u\in\mathcal{C}^{-\beta}(\mathbb{R}^d)$.
\end{proposition}
\begin{proof}
Let $\varphi\in\mathcal{D}(\mathbb{R}^d)$, $x\in\mathbb{R}^d$ and $\lambda\in(0,1]$. 
By definition of scaling degree,
\begin{align*}
\lvert u(\varphi^\lambda_x) \rvert \lesssim \lambda^{-sd_x(u)} \lesssim \lambda^{-\beta}\,,
\end{align*}
\textit{i.e.} $u\in\mathcal{C}^{-\beta}(\mathbb{R}^d)$.
\end{proof}
\vspace{2cm}

On behalf of all authors, the corresponding author states that there is no conflict of interest.

\end{document}